\documentclass[sigconf]{acmart}
\AtBeginDocument{%
  }

\setcopyright{acmlicensed}
\copyrightyear{2026}
\acmYear{2026}
\acmDOI{XXXXXXX.XXXXXXX}
\acmConference[Under Review]{}{}{}

\acmISBN{978-1-4503-XXXX-X/2018/06}

\usepackage{xcolor}
\usepackage{tcolorbox}
\usepackage{bm}
\usepackage{booktabs} 
\usepackage{multirow}
\usepackage{makecell}
\usepackage{graphicx}
\usepackage{subcaption}
\usepackage{float} 
\usepackage{hyperref}
\usepackage{xcolor} 
\definecolor{mypink}{RGB}{255, 105, 180} 
\usepackage{arydshln}

\theoremstyle{plain}
\newtheorem{theorem}{Theorem}[section]

\newtheorem{lemma}[theorem]{Lemma}

\theoremstyle{definition}

\newtheorem{assumption}[theorem]{Assumption}
\theoremstyle{remark}

\begin{document}

\title[RecGOAT]{RecGOAT: Graph Optimal Adaptive Transport for LLM-Enhanced Multimodal Recommendation with Dual Semantic Alignment}

\author{Yuecheng Li}
\affiliation{%
  \institution{Kuaishou Technology}
  \city{Beijing}
  \country{China}
}
\email{liyuecheng@kuaishou.com}

\author{Hengwei Ju}
\affiliation{%
  \institution{Fudan University}
  \city{Shanghai}
  \country{China}
}
\email{23210240199@m.fudan.edu.cn}

\author{Zeyu Song}
\affiliation{%
  \institution{Kuaishou Technology}
  \city{Beijing}
  \country{China}
}
\email{songzeyu@kuaishou.com}

\author{Wei Yang}
\affiliation{%
  \institution{University of Southern California}
  \city{Los Angeles}
  \country{USA}
}
\email{weiyangvia@gmail.com}

\author{Chi Lu}
\affiliation{%
  \institution{Kuaishou Technology}
  \city{Beijing}
  \country{China}
}
\email{luchi@kuaishou.com}

\author{Peng Jiang}
\affiliation{%
  \institution{Kuaishou Technology}
  \city{Beijing}
  \country{China}
}
\email{jiangpeng@kuaishou.com}

\author{Kun Gai}
\affiliation{%
  \institution{Unaffiliated}
  \city{Beijing}
  \country{China}
}
\email{gai.kun@qq.com}

\begin{abstract}
Integrating large language model (LLM) representations into multimodal recommendation has shown promise, yet a fundamental challenge remains largely overlooked: the semantic heterogeneity between generative LM representations and the ID-based collaborative signals that recommendation systems rely on. Naively injecting LM features without alignment degrades recommendation performance rather than improving it.
To resolve this, we propose \textbf{RecGOAT}, a dual-granularity semantic alignment framework built on graph neural networks and optimal transport theory. RecGOAT first enriches collaborative semantics through multimodal attentive graphs that capture item-item, user-item, and user-user relationships, initializing user representations via LLM-inferred behavioral preferences. It then aligns LM-derived modality representations with recommendation IDs at two complementary granularities: (1) instance-level alignment via cross-modal contrastive learning (CMCL), which produces discriminative per-sample representations; and (2) distribution-level alignment via optimal adaptive transport (OAT), which minimizes the 1-Wasserstein distance across modal distributions to produce a unified, consistently aligned feature space.
Theoretically, we prove that the unified representation achieves strictly lower target error than any single-modality representation, with the gap bounded by the Wasserstein distance and the InfoNCE loss—providing rigorous guarantees for both alignment consistency and fusion comprehensiveness. Extensive experiments on three public benchmarks demonstrate state-of-the-art performance. Deployment on a large-scale online advertising platform further validates RecGOAT's industrial scalability. Our code is available at \textcolor{blue}{\url{https://github.com/6lyc/RecGOAT-LLM4Rec}}.

\end{abstract}



\maketitle

\section{Introduction}

\begin{figure}[t]
    \centering
    \includegraphics[width=3.2in]{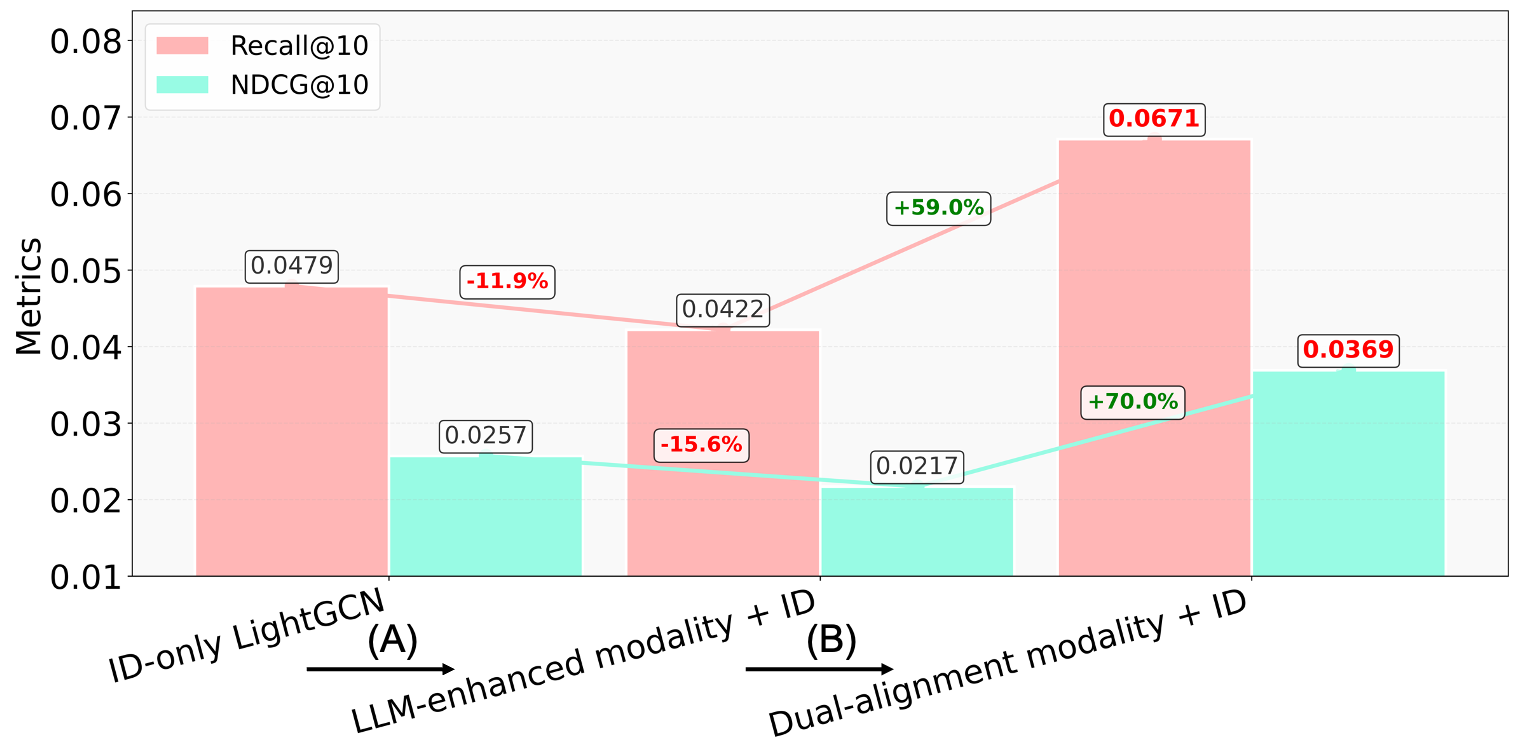}
    \caption{Performance comparison between LM representations with or without alignment for recommendation systems on Baby Dataset. (A) Due to semantic heterogeneity, LM Representation w/o alignment leads to degradation in recommendation performance. (B) Through our dual‑granularity alignment, the semantic conflict is resolved, yielding performance improvements of 59\% and 70\%, respectively.} 
    \label{g1}
\end{figure}

Recommendation systems (RS) have been widely adopted as essential filtering mechanisms in the era of information overload \cite{he2017NCF,wang2019NGCF, xia2022HCCF, zhao2024recommenderinllm, lin2025rsbllm}. However, the sparsity of explicit user-item interaction data severely constrains recommendation performance, particularly in large-scale recommendation scenarios. To address this, multimodal recommendation improves system performance by leveraging rich item content (such as textual descriptions and product images) to complement interaction signals, thereby alleviating data sparsity, enabling more accurate personalized recommendations, and enhancing the user experience \cite{chen2019VECF, zhou2023bm3, yang2024mami, yang2025fitmm,xu2025MMrecsurvey}.

Earlier studies primarily focused on employing convolutional neural networks (e.g., VGG, ResNet) and word embedding models (e.g., GloVe, BERT) to learn visual and textual modal information \cite{he2016vbpr, cui2018mvrnn, pomo2025VLRec}. These features were then fused with ID features via weighting, concatenation, or element-wise multiplication \cite{cui2018mvrnn, malitesta2025modalfusion}. To better capture higher-order interactions, many researchers have introduced graph neural networks (GNN) into multimodal recommendation, constructing user‑item and item‑item graphs to learn complex structural relations \cite{gao2023GNNRec, anand2025GNNRec}. However, due to constraints in network architecture and depth, the feature extraction capabilities of these models remain limited, and they lack sufficient semantic understanding of users and items. For example, they struggle to effectively initialize user ID embeddings, typically relying on random initialization \cite{he2020lightgcn, zhou2023freedom} or aggregating modal features from historically interacted items \cite{lin2024gume}. In practice, this means they fail to accurately depict user characteristics and genuine preferences, ultimately depending on intricate behavior modeling and supervised learning in downstream tasks.

Recently, several new architectures have been introduced into multimodal recommendation to address the challenge of deep semantic understanding, such as Transformer-based modal fusion frameworks \cite{li2023RecFormer, yi2024UGT}, Diffusion Model-based modal denoising and generative models \cite{yang2023DreamRec, jiang2024diffmm, li2025MoDiCF}, and Mamba-based efficient sequential modeling methods \cite{wang2025findrec}. With the scaling law being successfully validated across various fields, the latest research has further attempted to integrate LLMs, LVMs, and MLLMs with recommendation systems \cite{bao2023tallrec, wei2024UniMP, yi2025GollaRec}, aiming to leverage their rich world knowledge to enhance semantic representations of each modality. However, current LM-enhanced multimodal recommendation models still exhibit noticeable deficiencies in aligning modal signals with interaction IDs \cite{bao2023tallrec, wang2025IRLLRec}. We observe that \textbf{there exists a significant semantic heterogeneity between the world knowledge encapsulated in generative large models and the ID signals relied upon for user‑item interaction modeling}, as illustrated in Figure \ref{g1}. This heterogeneity hinders existing methods from fully unleashing the potential of large models in recommendation tasks. Therefore, \textbf{achieving thorough and consistent alignment} between the large models and recommendation systems has become key to pushing the performance ceiling of current recommendation systems.

To address the aforementioned challenges, this paper proposes a dual semantic alignment multimodal \textbf{\underline{Rec}}ommendation framework via \textbf{\underline{G}}raph \textbf{\underline{O}}ptimal \textbf{\underline{A}}daptive \textbf{\underline{T}}ransport (\textbf{RecGOAT}). 
\begin{itemize}
    \item  For \textbf{intra-modal learning}, we construct a collaborative signal representation enhancement module based on multimodal attentive graphs. By building multiple  graphs across text, image, and interaction with attention mechanisms, our approach strengthens multi-hop collaborative signals among item‑item, user‑item, and user‑user, thereby capturing high‑order structural relationships on both the user and item sides. To fully leverage the world knowledge and reasoning capabilities of large models, we employ Qwen3‑Embedding‑8B and LLaVA‑1.5‑7B to encode textual and visual features of items, respectively. Meanwhile, by constructing personalized behavioral prompts, we utilize QwQ‑32B to infer each user’s multi‑dimensional item preferences, which serve as initialized textual features for users. 
    \item For \textbf{cross‑modal alignment}, we design a dual-granularity semantic alignment framework between LLM-enhanced modalities and recommendation IDs. First, we perform instance‑level alignment via cross‑modal contrastive learning across text, vision, and ID modalities, obtaining discriminative multimodal representations. Second, we introduce an optimal adaptive transport technique to achieve distribution‑level alignment and representation fusion between semantic modalities and recommendation IDs. By minimizing the 1‑Wasserstein distance between different modal distributions, their feature embeddings are transported via an optimal transport matrix into a unified aligned space, yielding consistent and comprehensive fused item representations. Additionally, we incorporate adaptive learnable parameters into each transport matrix, which bridges OT alignment with the downstream recommendation task and enables precise supervision of the transport process. Notably, we provide a theoretical proof that the mutual constraint between any single‑source modal distribution and the unified fused distribution can be bounded by the Wasserstein distance and the InfoNCE loss. This demonstrates that the unified representations optimized in RecGOAT achieves strong \textbf{alignment consistency} and \textbf{fusion comprehensiveness}.
\end{itemize}

The main contributions of this paper are summarized as follows:

\begin{itemize}
    \item We propose RecGOAT, an LLM-enhanced multimodal recommendation framework that resolves the semantic heterogeneity between LLM-derived modality representations and ID-based collaborative signals. Our RecGOAT unifies structure-aware graph augmentation with a dual-granularity alignment objective, consisting of instance-level cross-modal contrastive learning (CMCL) and distribution-level optimal adaptive transport (OAT).
    \item We theoretically establish that the unified representation achieves a lower target error than any individual modality, with the error gap rigorously bounded by the Wasserstein distance and the InfoNCE loss, thereby offering guarantees of fusion comprehensiveness and alignment consistency.
    \item Extensive experiments on three public datasets and a large-scale online advertising platform demonstrate the effectiveness and scalability of RecGOAT. Ablations and analyses further confirm the necessity of OT-based distribution alignment and validate the claimed alignment consistency and fusion comprehensiveness.
\end{itemize}

\section{Related Work}
\subsection{Multimodal Recommendation}

Multimodal recommendation address the challenge of sparse user‑item interaction by extracting and integrating rich content features, such as textual, visual, and acoustic information. VBPR \cite{he2016vbpr} first extracted visual features via CNN and fused them with ID features through a weighted loss function. 
Given the strong capability of GNNs in capturing high‑order structural relations, MMGCN \cite{wei2019mmgcn} constructs user‑item bipartite graphs per modality and performs multi‑level message propagation. Further, LightGCN \cite{he2020lightgcn} simplifies graph convolution for collaborative filtering by retaining only the essential neighbor aggregation component. 
Moreover, LATTICE \cite{zhang2021LATTICE} and FREEDOM \cite{zhou2023freedom} construct dynamic/frozen item‑item graphs to further explore the potential of graph learning.

In recent years, inspired by the success of novel architectures such as Transformer \cite{vaswani2017attention} and Mamba \cite{gu2024mamba} in various fields, researchers have also introduced them into multimodal recommendation to enhance the representation of user preferences and item features. RecFormer \cite{li2023RecFormer} discards item ID dependency and addresses cold‑start and cross‑domain transfer problems with a bidirectional Transformer. UGT \cite{yi2024UGT} strengthens modality alignment and fusion through an end‑to‑end architecture combining multi‑way Transformer and a unified GNN. Furthermore, the Diffusion Model paradigm shifts recommendation from “classification” to “generation”: DreamRec \cite{yang2023DreamRec} generates oracle items via guided diffusion, avoiding negative sampling to eliminate noise interference; DiffMM \cite{jiang2024diffmm} extends multimodal adaptation by generating modality‑aware interaction graphs and incorporating cross‑modal contrastive learning. To improve model efficiency, FindRec \cite{wang2025findrec} employs linear‑complexity Mamba layers to model long‑range sequential dependencies, integrating Stein kernel distribution alignment with cross‑modal expert routing.

Despite notable progress in feature representation, scenario adaptation, and efficiency optimization achieved by the above multimodal recommendation methods, their core limitation lies in \textbf{their lack of the model scale required for deep semantic abstraction and reasoning}.

\subsection{LM-enhanced Recommendation}

Large models (LMs), empowered by their strong semantic understanding, cross‑modal integration, and knowledge transfer capabilities, have effectively compensated for key limitations of traditional recommendation systems, including inefficient modal fusion and insufficient fine‑grained preference modeling. They have thus gradually become a core driving force in advancing multimodal recommendation \cite{lopez2025LLM4MRSsurvey}. TALLRec \cite{bao2023tallrec} proposed an efficient two‑stage tuning framework, offering a foundational solution for adapting LLMs to recommendation scenarios. Rec‑GPT4V \cite{liu2024Rec-gpt4v} designed a Visual-Summary Thought strategy to convert item images into structured textual summaries. To address the issue that LLMs tend to overlook visual information in end‑to‑end fine‑tuning, NoteLLM‑2 \cite{zhang2025notellm2} introduced multimodal in‑context learning, contrastive learning, and a late‑fusion mechanism to balance attention across modalities. Differing from a single‑task focus, UniMP \cite{wei2024UniMP} constructed a unified multimodal personalization framework, integrating heterogeneous information through a unified data format and realizing multimodal alignment and fusion via cross‑layer cross‑attention. IRLLRec \cite{wang2025IRLLRec} focused on intent representation learning, employing dual‑tower alignment and momentum distillation to align textual intents with interaction intents. 


While the aforementioned LM‑enhanced recommendation methods have demonstrated strong performance, \textbf{their alignment mechanisms remain largely confined to instance‑level or pair‑wise local alignment, without considering global distribution‑level alignment}. As a result, these models struggle to capture global patterns across multimodal, cross‑domain, or interactive data.

\subsection{OT-based Distribution Alignment}

In the field of modal alignment, prior research has primarily relied on techniques such as supervised fine‑tuning \cite{bao2023tallrec}, contrastive learning \cite{zhang2025notellm2}, and graph‑structural learning (or cross‑attention mechanisms) \cite{yi2025GollaRec, wei2024UniMP}. However, these approaches are often confined to instance‑level or local node‑pair alignment, failing to account for the overall cross‑modal feature distribution. To overcome this limitation, we further introduce the concept of \textbf{distribution alignment}. Compared with commonly used KL divergence, the Wasserstein distance from optimal transport (OT) theory can directly characterize the geometric structure (e.g., shape and distance) between distribution supports, making it more suitable for matching and aligning modal distributions \cite{santambrogio2015OTapp, peyre2025OTML}. OTKGE \cite{cao2022otkge} proposed an optimal transport‑based knowledge graph embedding method, formulating multimodal fusion as an OT problem and optimizing the Wasserstein distance between distributions. Moreover, GOT \cite{chen2020got} constructs dynamic graphs and integrates the Wasserstein distance for node matching and the Gromov–Wasserstein distance for edge matching. In multimodal recommendation, MOTKD \cite{yang2023MOTKD} employs optimal transport to align textual, visual, and acoustic modalities, and designs a multi‑level knowledge distillation module to further strengthen alignment. However, \textbf{it does not explore the potential of OT for aligning modalities with IDs, nor does it provide theoretical guarantees for such alignment.}

\begin{figure*}[t]
\centering
\includegraphics[width=7.2in]{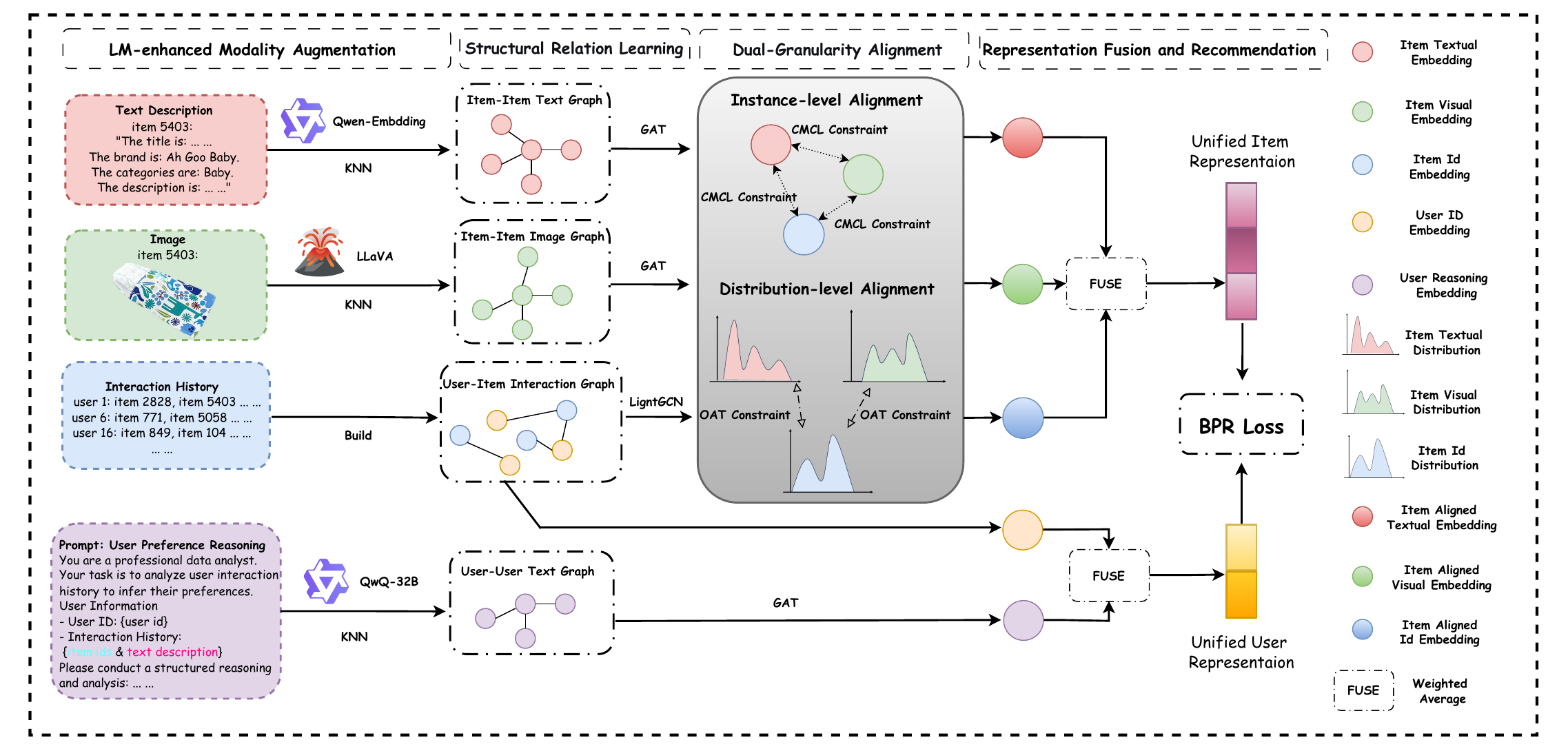}
\caption{The overall framework of our RecGOAT. It sequentially performs LM-enhanced modality augmentation for feature extraction, structural relation learning for graph-based collaborative signal modeling, dual-granularity alignment (instance and distribution levels) for cross-modal consistency, and final representation fusion \& recommendation.}
\label{g3}
\end{figure*}

\section{Methodology}

To address the semantic heterogeneity between modalities and IDs and fully unleash the potential of LMs in multimodal recommendation, we propose a novel dual semantic alignment framework, \textbf{RecGOAT}, which operates at both instance‑level and distribution‑level alignment. Section \ref{sec:3.1} and Section \ref{sec:3.2} describe the intra‑modal graph learning modeling with LM-enhanced modality and the cross‑modal dual semantic alignment framework, respectively, followed by the theoretical guarantees of our approach in Section \ref{sec:3.3}. In Section \ref{sec:3.4}, we employ the aligned unified representations for recommendation preference optimization. In Section \ref{sec:3.5}, we analyze the time and space complexities of the key modules in RecGOAT, demonstrating its potential for deployment in large-scale recommender systems. The overall architecture of RecGOAT is illustrated in Figure \ref{g3}.

\subsection{Intra‑modal: LM‑enhanced Modality Augmentation and Graph Learning}
\label{sec:3.1}
For each modality, we introduce large models (including LLMs and LVLMs) for enhancement and extract high‑order collaborative information by constructing attentive graphs.

\subsubsection{Item-Item Multimodal Graph Representation Learning}
\label{sec:3.1.1}
To enhance high‑order collaborative relationships between items, we construct separate textual and visual modality graphs \(\mathcal{G}^m = \{\mathcal{I}, \mathcal{E}^m, \mathcal{X}^m\}\), where \(m \in \{t, v\}\). Inspired by FREEDOM \cite{zhou2023freedom}, we adopt the K‑nearest neighbors (KNN) algorithm to build frozen item‑item graphs based on the initial LM‑enhanced modal features \(\mathcal{X}^m = \{\bm{x_i^m} \mid i=1,2,\dots,|\mathcal{I}|\}\). The LM-enhanced modal features \(\mathcal{X}^m\) are obtained by processing raw item content with pretrained large models. Specifically, for an item \(i\), its raw textual description \(t_i\) and visual image \(v_i\) are encoded separately:
    
\begin{equation}
\begin{aligned}
\bm{x_i^t} = f_{t}(t_{i} \mid \theta_{t}), \quad
\bm{x_i^v} = f_{v}(v_{i} \mid \theta_{v}),
\end{aligned}
\end{equation}
where \( f_{t} \) and \( f_{v} \) represent the pretrained LLM (e.g., Qwen-Embedding \cite{zhang2025qwen3embedding}) and LVLM (e.g., LLaVA \cite{liu2023llava}) parameterized by \( \theta_t \) and \( \theta_v \), respectively. For two items \(i, j \in \mathcal{I}^m\), the edge weight in the graph is computed via cosine similarity as follows: \( s_{ij}^m = \frac{\bm{x_i^m} \cdot \bm{x_j^m}}{\|\bm{x_i^m}\| \|\bm{x_j^m}\|}.
\)

To sparsify the modality graph, we retain only the top‑\(K\) edges with the highest similarity for each item node:

\begin{equation}
\mathcal{E}^m = \big\{ e_{ij}^m \mid i, j \in \mathcal{I}^m \big\}, \quad 
e_{ij}^m = 
\begin{cases} 
1, & \text{if } s_{ij}^m \in \operatorname{top-K}( \{s_{ik}^m \mid k \neq i\} ), \\
0, & \text{otherwise},
\end{cases}
\end{equation}
where \(e_{ij}^m = 1\) indicates the presence of an association edge between the two items.

After obtaining the graph \(\mathcal{G}^m\) for each modality, we employ Graph Attention Network \cite{velivckovic2018gat} to learn node representations \(\bm{z_i^m}\):





\begin{equation}
\bm{z_i^m} = \Bigg\|_{h=1}^{H} \sigma \left( \sum_{j \in \mathcal{N}_i^m} \alpha_{ij}^{m,h} \mathbf{W}^h \bm{x_j^m} \right) \in \mathbb{R}^d,
\end{equation}
where \(\alpha_{ij}^{m,h}\) is the normalized attention weight between node \(i\) and its neighbor \(j \in \mathcal{N}_i^m\) computed by the \(h\)-th attention head, \(\mathbf{W}^h\) is the transformation matrix corresponding to the \(h\)-th head. Ultimately, the representation of each item under modality \(m\) is enhanced by aggregating neighborhood information through multi‑head attention, thereby more effectively capturing high‑order collaborative signals within the modality.

\subsubsection{ID Embedding Learning from User-Item Interaction Graph}
\label{sec:3.1.2}
To capture collaborative filtering signals from implicit feedback, we construct a user-item interaction graph \(\mathcal{G}_{ui} = (\mathcal{U} \cup \mathcal{I}, \mathcal{E}_{ui})\), where \(\mathcal{U}\) and \(\mathcal{I}\) represent the sets of users and items, respectively. An edge \(e_{ui} \in \mathcal{E}_{ui}\) exists if user \(u\) has interacted with item \(i\). 

The ID embeddings for users and items are initialized as learnable parameters, denoted as \(\mathbf{E}^{(0)}_u \in \mathbb{R}^{|\mathcal{U}| \times d}\) and \(\mathbf{E}^{(0)}_i \in \mathbb{R}^{|\mathcal{I}| \times d}\), where \(d\) is the embedding dimension. Following the lightweight design of LightGCN \cite{he2020lightgcn}, the propagation rule at the \(l\)-th layer is defined as:

\begin{equation}
\mathbf{E}_u^{(l+1)} = \sum_{i \in \mathcal{N}_u} \frac{r_{ui}}{\sqrt{|\mathcal{N}_u|} \sqrt{|\mathcal{N}_i|}} \mathbf{E}_i^{(l)}, \quad
\mathbf{E}_i^{(l+1)} = \sum_{u \in \mathcal{N}_i} \frac{r_{ui}}{\sqrt{|\mathcal{N}_u|} \sqrt{|\mathcal{N}_i|}} \mathbf{E}_u^{(l)},
\end{equation}
where \(\mathcal{N}_u\) and \(\mathcal{N}_i\) are the neighbor sets of user \(u\) and item \(i\) in \(\mathcal{G}_{ui}\), respectively. To incorporate explicit preference signals available in certain datasets (e.g., review ratings in Amazon datasets), we introduce the rating value \(r_{ui}\) as attention coefficient between user \(u\) and item \(i\).

After \(L\) propagation layers, the final ID embeddings are obtained by averaging the representations from all layers: 

\begin{equation}
\bm{z^{id}_u} = \frac{1}{L+1} \sum_{l=0}^{L} \mathbf{E}_u^{(l)}, \quad
\bm{z^{id}_i} = \frac{1}{L+1} \sum_{l=0}^{L} \mathbf{E}_i^{(l)}.
\end{equation}

These refined ID embeddings encode high-order collaborative relations and are subsequently used for cross-modal alignment with LLM semantics.

\subsubsection{User-User Graph Learning with LLM Contextual Enhancement}
\label{sec:3.1.3}
Traditional ID-based recommendation models, compared to LLMs, lack persistent world knowledge and operate on coarse-grained IDs, which limits their generalization ability and understanding of various items. To fully leverage the world knowledge and reasoning capabilities of LLMs, we construct personalized behavioral prompts for each user based on their interaction history within graph \(\mathcal{G}_{ui}\) and the corresponding textual descriptions of interacted items. The prompt template is designed as follows: 

\tcbset{colback=gray!10, colframe=gray!88, boxrule=0.5pt, arc=2pt}

\begin{tcolorbox}[title=Prompt: User Preference Reasoning, fonttitle=\bfseries]
You are a professional data analyst. Your task is to analyze a user's interaction history to infer their preferences.


- User ID: \texttt{\{user id\}}
    
- Interaction History: \texttt{\{item ids \& text description\}}

Please conduct a structured reasoning by two steps:

- Identify Common Attributes Across Items: ... ...

- Summarize Preferences Across Multiple Dimensions: ... ...

Output Format: 

\texttt{<think> reasoning process here </think>}

\texttt{<answer> answer here </answer>}
\end{tcolorbox}
Next, the open‑source QwQ‑32B \cite{yang2025qwen3} is employed to infer user preferences based on this prompt. The generated answer within the \colorbox{gray!10}{\texttt{<answer>...</answer>}} are encoded into embeddings that serve as the user’s textual modal features. Let \(\mathcal{H}_u = \{(i, t_i) \mid i \in \mathcal{N}_u\}\) denote the set of interacted items and their textual descriptions for user \(u\), and let \(\mathcal{P}\) represent the structured prompt template. The LLM‑based preference reasoning and embedding generation are modeled as:

\begin{equation}
\begin{aligned}
a_u = f_{u}(\mathcal{H}_u \mid \mathcal{P}, \theta_{u}), \quad
\bm{{x}_u^t} = f_{t}(a_u \mid \theta_{t}),
\end{aligned}
\end{equation}
where \(a_u\) is the structured textual answer generated by the LLM parameterized by \(\theta_u\), and \({x}_u^t\) denotes the resulting textual modal feature obtained via the text encoder parameterized by \(\theta_t\).

Subsequently, a user-user textual modal graph \( \mathcal{G}^t_{uu} = (\mathcal{U}, \mathcal{E}^t_{uu}, \mathcal{X}^t_u) \) is constructed, where nodes represent users, features are the LLM-enhanced embeddings \( \mathcal{X}^t_u = \{\bm{{x}_u^t}\} \), and edges \( \mathcal{E}^t_{uu} \) are established based on the cosine similarity between user features, sparsified by retaining only the top-K connections for each node. Graph learning is then performed on \( \mathcal{G}^t_{uu} \) following the same multi-head graph attention network described in Section \ref{sec:3.1.1}. This allows the propagation and refinement of high-level, semantically enriched preferences among similar users, and the final refined user representation \( \bm{z_u^t} \) aggregates contextual signals from peers with semantically aligned preferences.

\subsection{Cross‑modal: Dual-Granularity Alignment of LLM-enhanced Modalities and ID Signals}
\label{sec:3.2}
To fully align LLM-enhanced semantic representations with recommendation ID signals, we propose a dual‑granularity alignment framework, consisting of instance‑level alignment based on cross‑modal contrastive learning and distribution‑level alignment based on optimal adaptive transport. These two components interact with and reinforce each other, jointly optimizing towards consistent and comprehensive item embeddings that organically integrate LLM semantics with interaction signals, thereby unleashing the full potential of LLM‑enhanced multimodal recommendation.

\subsubsection{Instance-level Alignment via \underline{C}ross-\underline{M}odal \underline{C}ontrastive \underline{L}earning (\textbf{CMCL})}
To achieve fine‑grained semantic alignment at the instance level, we perform cross‑modal contrastive learning that explicitly narrows the representation gap between different modalities of the same item while pushing apart those of different items. For each item \(i\), we sample paired representations from its available modalities: \textit{(ID, text)}, \textit{(ID, visual)}, and \textit{(text, visual)}. Each pair is treated as a positive example, while representations from different items within the same modality are considered negatives.

The contrastive objective is built upon the InfoNCE loss \cite{chen2020simcl}, which encourages the similarity between positive pairs to be higher than that between negative pairs. Formally, for a given anchor representation \(\bm{z_i^{m_a}}\) from modality \(m_a\) and its positive counterpart \(\bm{z_i^{m_p}}\) from modality \(m_p\), the contrastive loss for this pair is defined as:

\begin{equation}
\mathcal{L}_{i}^{(m_a, m_p)} = -\log \frac{\exp\big(\mathrm{sim}(\bm{z_i^{m_a}}, \bm{z_i^{m_p}}) / \tau\big)}{\sum_{j \in \mathcal{B}} \exp\big(\mathrm{sim}(\bm{z_i^{m_a}}, \bm{z_j^{m_p}}) / \tau\big)},
\end{equation}
where \(\mathrm{sim}(\cdot,\cdot)\) denotes cosine similarity, \(\tau\) is a temperature hyperparameter, and \(\mathcal{B}\) is the set of all items in the current batch that provide negative samples. The overall instance‑level cross-modal contrastive learning loss aggregates over all three modality pairs:

\begin{equation}
\mathcal{L}_{\text{CMCL}} = \frac{1}{B}\sum^{B}_{i=1} \sum_{(m_a, m_p) \in \mathcal{MP}} \mathcal{L}_{i}^{(m_a, m_p)},
\end{equation}
with \(\mathcal{MP} = \{(id, t), (id, v), (t, v)\}\).

The representations refined through this contrastive learning process are not only semantically discriminative within and across modalities, but also provide a meaningful similarity structure that reflects genuine semantic relatedness \cite{li2024cdnmf}. In the next section, these aligned representations are used to compute the cost matrix in the Optimal Transport problem. This ensures that the transport cost between two items captures their deep semantic discrepancy, thereby guiding the OT to perform semantically‑aware distribution alignment rather than relying solely on raw feature distances.

\subsubsection{Distribution-level Alignment via \underline{O}ptimal \underline{A}daptive \underline{T}ransport (\textbf{OAT})}

Semantic heterogeneity at the distribution level undermines the efficacy of large models in multimodal recommendation and caps the performance ceiling of existing systems. To achieve principled alignment between LLM‑enhanced modal representations and recommendation ID embeddings, we formulate the semantic alignment process as an Optimal Transport (OT) problem \cite{peyre2019computationalOT}. Specifically, we aim to transport the LLM‑augmented semantic feature distribution (\textit{source}) to match the collaborative ID feature distribution (\textit{target}), which naturally quantifies and minimizes the distributional divergence between heterogeneous semantic spaces.

Formally, let \(P^m\) denote the empirical distribution of the LLM-enhanced modality \(m\) (where \(m \in \{t, v\}\)), and \(Q^{id}\) denote the empirical distribution of the ID embeddings obtained from Section \ref{sec:3.1.2}. The OT problem seeks a coupling \(t\) that minimizes the total cost of moving mass from \(P^m\) to \(Q^{id}\). In its continuous form, this is expressed as:

\begin{equation}
OT(P^m, Q^{id}) = \inf_{t \in \Pi(P^m, Q^{id})} \int_{\mathcal{Z}^m \times \mathcal{Z}^{id}} c(\bm{z^m}, \bm{z^{id}}) \, d \, t(\bm{z^m}, \bm{z^{id}}),
\end{equation}
where \(\Pi(P^m, Q^{id})\) is the set of all joint distributions with marginals \(P^m\) and \(Q^{id}\), and \(c: \mathcal{Z}^m \times \mathcal{Z}^{id} \rightarrow \mathbb{R}^+\) is a cost function measuring the semantic dissimilarity between a source feature \(\bm{z^m}\) and a target feature \(\bm{z^{id}}\).

To concretely quantify the gap between LLM semantics and ID-based collaborative signals, we define the \textbf{feature-wise cost} as the normalized $L_1$ distance between feature distributions. For a batch of \(B\) samples, let \(\bm{Z^m} \in \mathbb{R}^{B \times d}\) represent the LLM-enhanced features from modality \(m\) and \(\bm{Z^{id}} \in \mathbb{R}^{B \times d}\) represent the ID embeddings. The cost matrix \( \bm{C^m} \in \mathbb{R}^{d \times d} \) is computed as:

\begin{equation}
\mathbf{C}^m_{ij} = s \cdot \frac{1}{B} \sum_{b=1}^{B} \left\| \bm{Z_{b,i}^m} - \bm{Z_{b,j}^{id}} \right\|_1,
\end{equation}
where $s$ is scaling factor used to adjust the scale of the cost matrix and ensure numerical stability.

Let \(P^m = \frac{1}{B}\sum_{i=1}^B \delta_{\bm{z_i^m}}\), and \(Q^{id} = \frac{1}{B}\sum_{i=1}^B \delta_{\bm{z^{id}_i}}\), where $\delta(\cdot)$ is the Dirac delta function. Then, the discrete OT problem then reduces to minimizing the 1‑Wasserstein distance \(\mathcal{W}_1\)  between the two distributions:
\begin{equation}
\mathcal{W}_1(P^m, Q^{id}) = \min_{\bm{T} \in \Pi(\bm{p}, \bm{q})} \langle \bm{T}, \bm{C}^m \rangle_F =  \min_{\bm{T} \in \Pi(\bm{p}, \bm{q})} \sum_{i=1}^{d} \sum_{j=1}^{d} T_{ij} C^m_{ij},
\end{equation}
subject to \(\bm{T} \mathbf{1} = \bm{p}, \bm{T}^\top \bm{1} = \bm{q}\), where $\bm{1}$ is the all-one vector. Here, \(\bm{T} \in \mathbb{R}^{d \times d} \) is the transport plan matrix, \(\bm{p}\) and \(\bm{q}\) are uniform weight vectors, and \(\langle \cdot, \cdot \rangle_F\) denotes the Frobenius inner product.

We solve this entropy‑regularized OT problem efficiently using the Sinkhorn‑Knopp algorithm \cite{sinkhorn1967first, cuturi2013sinkhorn}, which iteratively updates row and column scaling vectors to converge linearly to an approximate optimal transport plan \(\mathbf{T}_0^m\) for each modality $m$.

To enable the OT alignment to adapt to downstream recommendation tasks, we augment the base OT plan with a learnable residual matrix \(\widetilde{\mathbf{T}}^m\). The final \textbf{adaptive transport optimal matrix} for modality \(m\) is:
\begin{equation}
\bm{T}^m = \bm{T}_0^m + \widetilde{\bm{T}}^m.
\end{equation}
This allows the model to fine‑tune the purely geometry‑driven coupling \(\bm{T}_0^m\) with task‑specific semantic corrections. Using \(\mathbf{T}^m\), each LLM-enhanced modal feature is transported toward the ID embedding space via:
\begin{equation}
\bm{\hat{Z}^m} =  \bm{Z^m} \cdot \bm{T^m}.
\end{equation}

Based on the distribution-level alignment process described above, we have effectively \textbf{mitigated the semantic heterogeneity between the LLM-enhanced modal spaces and the ID-based collaborative space}. This yields three semantically aligned item representations: \(\hat{\mathbf{Z}}^t\), \(\hat{\mathbf{Z}}^v\), and \(\mathbf{Z}^{id}\). While each of these representations captures consistent semantic information from its respective space, a comprehensive item embedding must integrate complementary cues from all available modalities. To this end, we fuse the three aligned representations via a weighted averaging:

\begin{equation} \label{equation:16}
\bm{Z} = \gamma_t \cdot \hat{\bm{Z}}^t + \gamma_v \cdot \hat{\bm{Z}}^v + \bigl(1 - \gamma_t - \gamma_v \bigr) \cdot \bm{Z}^{id},
\end{equation}
where \(\gamma_t, \gamma_v \in [0, 1]\) are hyperparameters, which determine the relative contribution of each aligned modality to the final unified representation.

The unified representation \(\mathbf{Z}\) thus embodies both semantic consistency, inherited from the distribution-aligned features, and informational comprehensiveness, achieved by combining multimodal and collaborative views. It serves as the final item embedding for downstream preference prediction, seamlessly \textbf{bridging the rich, world‑aware semantics from large models with the interaction‑driven relational knowledge from GNN‑based ID representations}.

\subsection{Theoretical Guarantees for Alignment Consistency and Fusion Comprehensiveness}
\label{sec:3.3}
To rigorously analyze the performance of RecGOAT, we provide mathematical proofs for its alignment consistency and fusion comprehensiveness, which are promoted by the joint optimization of instance-level and distribution-level alignment losses. The analysis focuses on the item side while treating user embeddings as fixed, under well-defined and realistic assumptions \cite{courty2017OTproof}.

\subsubsection{Problem Setup and Assumptions}
Let \( \bm{U} \) be the set of fixed user embeddings. For any \( \bm{u} \in \bm{U} \), we assume \( \|\bm{u}\| \leq K \) for a constant \( K > 0 \). The true preference function is denoted by \( f^*(\bm{u}, \bm{v}) \), which maps a user \( \bm{u} \) and an item \( \bm{v} \) to a real-valued score. Our model's rating function is the inner product \( f(\bm{u}, \bm{z}) = \bm{u}^\top \bm{z} \), where \( \bm{z} \) is an item representation. Let \( Q \) be the distribution of the unified item representations. We define the modality-specific error and the unified representation error as:
\begin{equation*}
\epsilon_m(f) = \mathbb{E}_{z \sim P^m} \bigl[ | f(u, z) - f^*(u, v) | \bigr],
\epsilon_F(f) = \mathbb{E}_{z \sim Q} \bigl[ | f(u, z) - f^*(u, v) | \bigr].
\end{equation*}

To facilitate the derivation, we adopt the following reasonable assumptions \cite{cao2022otkge}.
\begin{assumption}[Bounded User Embeddings]
\label{assumption:1}
 All user embeddings are fixed and bounded, i.e., \( \|\bm{u}\| \leq K \). Consequently, for a fixed \( \bm{u} \), the scoring function \( f(\bm{u}, \bm{z}) = \bm{u}^\top \bm{z} \) is \( K \)-Lipschitz continuous w.r.t \( \bm{z} \): \( \left| f(\bm{u}, \bm{z_1}) - f(\bm{u}, \bm{z_2}) \right| \leq K \cdot \|\bm{z_1} - \bm{z_2}\| \). This follows directly from the Cauchy-Schwarz inequality: \( |u^\top(\bm{z_1} - \bm{z_2})| \leq \|\bm{u}\| \cdot \|\bm{z_1} - \bm{z_2}\| \leq K \cdot \|\bm{z_1} - \bm{z_2}\| \).
\end{assumption}

\begin{assumption}[Lipschitz Continuity of True Preference]
\label{assumption:2}
    The true preference function \( f^*(\bm{u}, \bm{v}) \) is \( L^* \)-Lipschitz continuous with respect to the item representation \( \bm{z} \): \( |f^*(\bm{u}, \bm{z_1}) - f^*(\bm{u}, \bm{z_2})| \leq L^* \|\bm{z_1} - \bm{z_2}\| \). This reflects the inherent smoothness of user preferences.
\end{assumption}

\subsubsection{Supporting Lemmas and  Main Theorem} \label{sec:3.3.2}
Building upon the preceding definitions and assumptions, two key lemmas are introduced, which subsequently lead to the theorem and proof concerning alignment consistency and fusion comprehensiveness.


\begin{lemma}[Instance-level Distance Bound]
\label{lemma:1}
Let \( \bm{z}_i^m \) and \( \bm{z}_i \) be the \( L_2 \)-normalized representations for modality \( m \) and the unified representation for item \( i \) (\( i.e., \|\bm{z}_i^m\|_2 = 1, \|\bm{z}_i\|_2 = 1 \)), respectively. The expected pairwise distance is bounded by the contrastive loss:
\begin{equation}
\mathbb{E}_i\!\bigl[\|\bm{z}_i^m - \bm{z}_i\|_2\bigr] \le \sqrt{ 4 + 2\tau \ln\left( \frac{2\mathcal{L}_{CMCL}}{B-1} \right) }
\end{equation}
where \( \tau > 0 \) is the temperature parameter, \( B > 1\) is the batch size, and \( \mathcal{L}_{CMCL} \) is the global cross-modal contrastive loss.
\end{lemma}

\begin{lemma}[Modality-to-Unified Error Bound]
\label{lemma:2}
For any modality \( m \) and any fixed user \( \bm{u} \), the difference between the modality-specific error and the unified error is bounded by both distributional and instance-level alignment terms:
\begin{equation}
\left| \epsilon_m(f) - \epsilon_F(f) \right| \leq (K + L^*) \cdot \mathcal{W}_1(P^m, Q^{id}) + K \cdot \mathbb{E}_i \| \bm{z_i^m} - \bm{z_i} \|,
\end{equation}
where \( \mathcal{W}_1(P^m, Q^{id}) \) is the 1-Wasserstein distance between the modality distribution \( P^m \) and the ID-based target distribution \( Q^{id} \).
\end{lemma}
The right-hand side of the inequality is bounded by the distribution-level Wasserstein distance and the instance-level Euclidean distance. This lemma formally justifies the rationality of our dual-granularity semantic alignment. Next, we give the Theorem \ref{theorem: 1} and its proof.

\begin{theorem}[Alignment Consistency and Fusion Comprehensiveness of RecGOAT]
\label{theorem: 1}
For any fixed user embedding \( \bm{u} \in \mathcal{U} \) and for all modalities \( m \in \mathcal{M} = \{t, v, \text{id}\} \), the following guarantees hold:

(1) \textbf{Consistency Guarantee}:
\begin{equation} \label{equation: 19}
\begin{aligned}
\max_{m \in \mathcal{M}} \, \epsilon_m(f)  - \epsilon_F(f) \ \leq \ & (K + L^*) \cdot \mathcal{W}_1(P^m, Q^{id}) \\
& + \sqrt{ 4 K^2 + 2\tau K^2 \ln\left( \frac{2\mathcal{L}_{CMCL}}{B-1} \right) }.
\end{aligned}
\end{equation}

(2) \textbf{Comprehensiveness Guarantee}:
\begin{equation} \label{equation:20}
\begin{aligned}
\epsilon_F(f) \ \leq \ \min_{m \in \mathcal{M}} \Bigl\{ & \epsilon_m(f) + (K + L^*) \cdot \mathcal{W}_1(P^m, Q^{id}) \\
& + \sqrt{ 4 K^2 + 2\tau K^2 \ln\left( \frac{2\mathcal{L}_{CMCL}}{B-1} \right) } \Bigr\}.
\end{aligned}
\end{equation}
\end{theorem}
\begin{proof}
Starting from Lemma \ref{lemma:2}, we have for any modality \( m \):
\begin{equation*}
\epsilon_m(f) - \epsilon_F(f) \leq (K + L^*) \cdot \mathcal{W}_1(P^m, Q^{id}) + K \cdot \mathbb{E}_i \| \bm{z_i^m} - \bm{z_i} \|.
\end{equation*}
Applying the bound from Lemma \ref{lemma:1} to instance-level term yields:
\begin{equation*}
\epsilon_m(f) \leq \epsilon_F(f) + (K + L^*) \cdot \mathcal{W}_1(P^m, Q^{id}) + K \cdot \sqrt{ 4 + 2\tau \ln\left( \frac{2\mathcal{L}_{CMCL}}{B-1} \right) }.
\end{equation*}
Since this inequality holds for all \( m \in \mathcal{M} \), taking the maximum over modalities on the left side yields the \textbf{Consistency Guarantee}:
\small \[
\max_{m \in M} \epsilon_m(u) \leq \epsilon_F(u) + (K + L^*) \cdot \mathcal{W}_1(P^m, Q^{id}) + \sqrt{ 4 K^2 + 2\tau K^2 \ln\left( \frac{2\mathcal{L}_{CMCL}}{B-1} \right) }.
\]


Similarly, from Lemma \ref{lemma:2} we also have:
\begin{equation*}
\epsilon_F(f) - \epsilon_m(f) \leq (K + L^*) \cdot \mathcal{W}_1(P^m, Q^{id}) + K \cdot \mathbb{E}_i \| \bm{z_i^m} - \bm{z_i} \|.
\end{equation*}
Applying the same substitution from Lemma \ref{lemma:1} gives:
\begin{equation*}
\epsilon_F(u) \leq \epsilon_m(u) + (K + L^*) \cdot \mathcal{W}_1(P^m, Q^{id}) + K \cdot \sqrt{ 4 + 2\tau \ln\left( \frac{2\mathcal{L}_{CMCL}}{B-1} \right) }.
\end{equation*}
As this is valid for all \( m \in M \), the tightest bound is achieved by taking the minimum over modalities on the right-hand side, resulting in the \textbf{Comprehensiveness Guarantee}:
{\small \begin{equation*}
    \epsilon_F(u) \leq \min_{m \in M} \left\{ \epsilon_m(u) + (K + L^*) \cdot \mathcal{W}_1(P^m, Q^{id}) + \sqrt{ 4 K^2 + 2\tau K^2 \ln\left( \frac{2\mathcal{L}_{CMCL}}{B-1} \right) } \right\}.
\end{equation*}}
\end{proof}
Based on Theorem \ref{theorem: 1}, Eq. (\ref{equation: 19}) indicates that by optimizing the Wasserstein distance \( \mathcal{W}_1(P^m, Q^{id}) \), while reducing the contrastive learning loss \( \mathcal{L}_{\text{CMCL}} \), the recommendation consistency between the modal representations with the largest error and the unified representation can be better aligned. On the other hand, Eq. (\ref{equation:20}) shows that the error of the fused unified representation does not exceed the error of any single modality plus the dual-granularity alignment error. That is, through OT-based distribution alignment and cross-modal contrastive learning, the fused representation can effectively integrate multimodal information and enhance recommendation performance. This theorem provides theoretical assurance and principled foundation for our RecGOAT: through dual-granularity alignment (instance-level contrastive learning + distribution-level OT mapping), the model successfully bridges the semantic gap between LLM-enhanced modalities and ID-based interaction signals to achieve both consistent and comprehensive multimodal fusion.

\subsection{Preference Optimization for Recommender}
\label{sec:3.4}
In summary, we optimize the downstream recommendation task using the Bayesian Personalized Ranking (BPR) loss \cite{rendle2009bpr}. The fused user representation $\bm{U}$ is obtained by weighting the user ID embedding \(\bm{z_u^{id}}\) from Section \ref{sec:3.1.2} and the enhanced textual user representation \(\bm{z_u^{t}}\) from Section \ref{sec:3.1.3}. Together with the unified item representation $\bm{Z}$ from Eq. (\ref{equation:16}), the model is optimized with the BPR loss as follows:

\begin{equation}
    \mathcal{L}_{\text{BPR}} = \sum_{(u, i, j) \in \mathcal{O}} -\ln \sigma(f(\bm{u}, \bm{z_i}) - f(\bm{u}, \bm{z_j})),
\end{equation}
where \(\mathcal{O}\) denotes the set of observed \((user, positive\ item, negative\ item)\) triplets, \(\sigma\) is the sigmoid function, and the scoring function is defined as \(f(u, z) = u^\top z\). Items with higher predicted scores are ranked as high-potential candidates for recommendation.

\subsection{Complexity Analysis}
\label{sec:3.5}
To improve the scalability of RecGOAT for large-scale recommendations, our OAT module aligns feature distributions (relying on feature dimension $d$) rather than performing standard OT node matching (which heavily relies on sample size $N$). Consequently, by employing the Sinkhorn-Knopp algorithm (with a maximum of $L$ iterations), the training time and space complexities are reduced to $\mathcal{O}((N+L) \cdot d^2)$ and $\mathcal{O}(N \cdot d^2)$ respectively, achieving linear scalability with respect to the sample size.

\begin{table*}[t]
\renewcommand{\arraystretch}{1.6}
\centering
\caption{Recommendation performance on three Amazon Datasets. Here, R@10 and N@10 denote Recall@10 and NDCG@10, respectively. The best results are highlighted in \textbf{\textit{bold}}, and the second-best are \underline{\textit{underlined}}. The \textit{asterisk}* indicates that the improvement of our RecGOAT is statistically significant based on t-test with $p$-value $< 0.001$. Our model achieves statistically significant state-of-the-art performance on all metrics across each dataset.}
\label{tab:main_results}
\resizebox{\textwidth}{!}{
\begin{tabular}{lcccccccccccccc}
\toprule
\multirow{3}{*}{\textbf{Dataset}} & \multirow{3}{*}{\textbf{Metric}} 
& \multicolumn{2}{c}{\textbf{ID-based Methods}} 
& \multicolumn{5}{c}{\textbf{Multimodal Methods}} 
& \multicolumn{4}{c}{\textbf{Large Models-based Methods}} 
& \multicolumn{2}{c}{\textbf{Ours}} \\
\cmidrule(lr){3-4}
\cmidrule(lr){5-9}
\cmidrule(lr){10-13}
\cmidrule(lr){14-15}
& & BPR & LightGCN & VBPR & FREEDOM & DiffMM 
& UGT & FindRec & TALLRec & A-LLMRec 
& UniMP & IRLLRec 
& \multirow{2}{*}{\textbf{RecGOAT}} & \multirow{2}{*}{\textbf{Improv.}} \\
& & (UAI'09) & (SIGIR'20) & (AAAI'16) & (MM'23) & (MM'24) 
& (RecSys'24) & (KDD'25) & (RecSys'23) & (KDD'24) 
& (ICLR'24) & (SIGIR'25)  \\
\midrule
\multirow{2}{*}{Baby} 
& R@10  & 0.0357 &	0.0479 & 0.0423 & 0.0624 & 0.0617 & 0.0602 & \underline{0.0647} &	0.0382 & 0.0379 & 0.0472 & 0.0624 & \textbf{0.0671*} &	\textcolor{mypink}{$\uparrow$ 3.71\%} \\
& N@10  & 0.0192 & 0.0257 & 0.0223 & 0.0324 & 0.0321 & 0.0325 & \underline{0.0348}	& 0.0197 & 0.0203 & 0.0267 & 0.0318 & \textbf{0.0369*} & \textcolor{mypink}{$\uparrow$ 6.03\%} \\
\midrule
\multirow{2}{*}{Sports} 
& R@10  & 0.0432 & 0.0569 & 0.0558 & 0.0710 & 0.0687 & 0.0705 & 0.0707  & 	0.0418 & 0.0402 & 0.0528 & \underline{0.0712}  & \textbf{0.0745*} & \textcolor{mypink}{$\uparrow$ 4.63\%} \\
& N@10  & 0.0241 & 0.0311	& 0.0307 	& 0.0382	& 0.0357	& \underline{0.0391} & 0.0383 & 0.0247 & 0.0223 & 0.0288 & 0.0375 & \textbf{0.0415*} & \textcolor{mypink}{$\uparrow$ 6.14\%} \\
\midrule
\multirow{2}{*}{Electronics} 
& R@10  & 0.0235 & 0.0363	& 0.0293	& 0.0396	& 0.0386 & \underline{0.0430} & 0.0395	& 0.0374 & 0.0347 &  0.0363 & 0.0419 & \textbf{0.0468*} & \textcolor{mypink}{$\uparrow$ 8.84\%} \\
& N@10  & 0.0127	& 0.0204	& 0.0159 	& 0.0220	& 0.0228 & \underline{0.0254} & 0.0210 & 0.0178 & 0.0201	& 0.0215 & 0.0248 & \textbf{0.0271*} & \textcolor{mypink}{$\uparrow$ 6.69\%} \\
\bottomrule
\end{tabular}
}
\end{table*}

\section{Experiments}

In this section, we conduct extensive experiments on three public Amazon datasets and a large-scale online advertising platform to address the following key research questions:

\begin{itemize}
    \item \textbf{RQ1}: Does our RecGOAT achieve SOTA performance compared to classical recommendation methods as well as leading multimodal and large model-based approaches? 
    \item \textbf{RQ2}: What is the negative impact of semantic conflict between LLM-enhanced modalities and ID signals? What are the individual and combined contributions of Cross-Modal Contrastive Learning (CMCL) and Optimal Adaptive Transport (OAT) in resolving semantic heterogeneity and improving recommendation performance?
    \item \textbf{RQ3}: How does our RecGOAT demonstrate alignment consistency and fusion comprehensiveness?
    \item \textbf{RQ4}: How effective are the key modules of RecGOAT in improving performance for an industrial-level online advertising system?
    \item \textbf{RQ5}: How sensitive is our RecGOAT to hyperparameter settings?
    \item \textbf{RQ6}: What is the computational efficiency of RecGOAT compared to other advanced methods?
\end{itemize}

\subsection{Experimental Setup}
\subsubsection{Datasets} We conduct experiments on three public Amazon datasets \footnote{\url{https://cseweb.ucsd.edu/~jmcauley/datasets/amazon/links.html}}: Baby, Sports, and Electronics \cite{mcauley2015amazon}. Each dataset contains user-item interactions along with visual and textual descriptions of items. Detailed statistics of datasets are summarized in Table \ref{tab:dataset_stats}.

\begin{table}[t]
\centering
\renewcommand{\arraystretch}{1.0}
\caption{Statistics of our experimental datasets.}
\label{tab:dataset_stats}
\begin{tabular}{lcccc}
\toprule
Datasets & \# Users & \# Items & \# Interactions & Sparsity \\
\midrule
Baby & 19,445 & 7,050 & 160,792 & 99.88\% \\
Sports & 35,598 & 18,357 & 296,337 & 99.95\% \\
Electronics & 192,403 & 63,001 & 1,689,188 & 99.99\% \\
\bottomrule
\end{tabular}
\end{table}

\subsubsection{Baselines and Evaluation Metrics}
We compare our RecGOAT with the following three categories of representative multimodal recommendation methods: (1) \textbf{Traditional ID-based Methods}: BPR \cite{rendle2009bpr} and LightGCN \cite{he2020lightgcn}. (2) \textbf{Multimodal Methods}: VBPR \cite{he2016vbpr} (CNN-based), FREEDOM \cite{zhou2023freedom} (GNN-based), DiffMM \cite{jiang2024diffmm} (Diffusion-based), UGT \cite{yi2024UGT} (Transformer-based), and FindRec \cite{wang2025findrec} (Mamba-based). (3) \textbf{LM-enhanced Methods} (with different semantic alignment/fusion paradigms): TALLRec \cite{bao2023tallrec} (fine-tuning), A-LLMRec \cite{kim2024allmrec} (in-context learning), UniMP \cite{wei2024UniMP} (cross-attention), and IRLLRec \cite{wang2025IRLLRec} (contrastive learning + KL divergence).

To evaluate the recommendation performance, we adopt two widely-used metrics: Recall (R@K) and Normalized Discounted Cumulative Gain (NDCG, N@K), where K is set to 10. Each metric is computed over 10 runs, and the average result is reported.

\subsubsection{Implementation Details} Following common practice \cite{zhou2023bm3, xu2025fastMMrec}, we split each dataset into an 8:1:1 ratio for training, validation, and testing under the 5-core setting. For multimodal baselines, we uniformly employ the publicly available 4096-dimensional visual features and 384-dimensional textual features provided by the open-source framework MMRec \cite{zhou2023mmrec}, adhering to its standard parameter configuration. For LM-enhanced Methods, we consistently apply the LLM-enhanced modality inputs introduced in this work.


\subsection{Overall Performance (RQ1)}

To verify the core motivation of this paper and demonstrate the advancement of RecGOAT, we compare it with three representative categories of multimodal recommendation methods, as presented in Table \ref{tab:main_results}, and draw the following key conclusions:
\begin{enumerate}
    \item Multimodal methods (e.g., FindRec, FREEDOM) outperform traditional ID-based methods (e.g., LightGCN), confirming the auxiliary role of multimodal information in alleviating the sparsity of interaction IDs. However, LLM-enhanced methods, despite their powerful semantic extraction and generation capabilities, generally underperform multimodal baselines while incurring significantly higher computational costs. This performance gap stems from their insufficient emphasis on aligning world knowledge with recommendation ID signals.
    \item Our RecGOAT achieves substantial improvements over LLM-enhanced baselines (e.g., 0.0468 vs. 0.0419 for Electronics), primarily due to its theoretically-grounded dual semantic alignment, especially the previously overlooked principle of distribution-level alignment between LLM-enhanced modalities and ID embeddings. Specifically, RecGOAT’s superior performance over TALLRec (fine tuning-based alignment), A‑LLMRec (in-context learning‑based alignment), and UniMP (cross‑attention‑based alignment) demonstrates that alignment from a distribution perspective is crucial, and instance‑level or pair‑wise alignment alone is insufficient. 
    \item Our RecGOAT outperforms IRLLRec (contrastive learning + KL divergence‑based alignment), indicating the advantage of the Wasserstein distance (OT) over KL divergence for distribution alignment. KL divergence only measures the ratio of probability densities and is insensitive to the geometric structure of the sample space. For example, aligning the feature “red” as “purple” incurs a similar penalty as aligning it as “indoor item” under KL divergence, whereas the Wasserstein distance would assign a much lower cost to the former, more semantically related distribution.
\end{enumerate}

Overall, by integrating cross‑modal contrastive learning and optimal adaptive transport within a dual‑alignment framework, RecGOAT achieves state‑of‑the‑art recommendation performance.

\begin{table}[t]
\renewcommand{\arraystretch}{1.1}
\setlength{\tabcolsep}{1.5pt}
\centering
\small 
\caption{Ablation study on different alignment and fusion strategies for three Amazon Datasets.}
\label{tab:ablation_study}
{
\begin{tabular}{@{}lccccccc@{}}
\toprule
\multirow{2}{*}{\textbf{Dataset}} & \multirow{2}{*}{\textbf{Metric}} 
& \multicolumn{1}{c}{\textbf{ID-only}} 
& \multicolumn{2}{c}{\textbf{Naive MM Fusion}} 
& \multicolumn{2}{c}{\textbf{Our Alignment}} 
& \multicolumn{1}{c}{\textbf{Ours}} \\
\cmidrule(lr){3-3}
\cmidrule(lr){4-5}
\cmidrule(lr){6-7}
\cmidrule(lr){8-8}
& & LightGCN & Concat & Sum & w/ CMCL & w/ OAT & {\textbf{RecGOAT}}  \\
\midrule
\multirow{2}{*}{Baby} 
& R@10  & 0.0479 & 0.0472	& 0.0422	& 0.0601	& \underline{0.0623}	& \textbf{0.0671*} \\
& N@10  & 0.0257 & 0.0244	& 0.0217	& 0.0330	& \underline{0.0346}	& \textbf{0.0369*} \\
\midrule
\multirow{2}{*}{Sports} 
& R@10  & 0.0569	& 0.0573	& 0.0525	& 0.0695	& \underline{0.0718}	& \textbf{0.0745*} \\
& N@10  & 0.0311	& 0.0305	& 0.0277	& 0.0381	& \underline{0.0397}	& \textbf{0.0415*} \\
\midrule
\multirow{2}{*}{Elec.} 
& R@10  & 0.0363	& 0.0402	& 0.0385	& 0.0388	& \underline{0.0437}	& \textbf{0.0468*} \\
& N@10  & 0.0204	& 0.0222	& 0.0210	& 0.0217	& \underline{0.0245}	& \textbf{0.0271*} \\
\bottomrule
\end{tabular}
}
\end{table}

\subsection{Ablation Study (RQ2)}

To demonstrate the semantic conflict between LLM-enhanced modalities and ID signals and to quantify the effectiveness of different alignment strategies, we evaluate several variants: an ID-only method (LightGCN), naive Multimodal (MM) fusion (i.e., Concat and Sum) with ID from LightGCN and modality form GAT, and our RecGOAT with individual or combined alignment components (i.e., CMCL and OAT). The results are summarized in Table \ref{tab:ablation_study}, leading to the following observations:

\begin{itemize}
    \item Simple fusion of LLM-enhanced modal embeddings via concatenation or summation yields inferior or inconsistent performance compared to the ID-only LightGCN (e.g., on the Baby dataset), confirming the severe semantic heterogeneity between large model semantics and recommendation IDs.
    \item  Within our dual-granularity alignment framework, OAT consistently outperforms CMCL across all datasets, highlighting the critical role of distribution-level alignment. Furthermore, the organic integration of instance-level and distribution-level alignment mutually reinforces both components, resulting in comprehensive semantic fusion and optimal multimodal recommendation performance.
\end{itemize}

\subsection{Alignment Consistency and Fusion Comprehensiveness (RQ3)}
To validate the theoretical conclusions established in Section \ref{sec:3.3}, we conducted experiments on the Baby dataset to examine alignment consistency and fusion comprehensiveness, as illustrated in Figure \ref{fig:g3}. First, Figure (\ref{fig:g_3_1}) demonstrates that the impact of different weighting coefficients in Eq. (\ref{equation:16}) on the final recommendation performance is robust, indicating strong consistency among the aligned representations $\hat{\bm{Z}}^t$, $\hat{\bm{Z}}^v$, and $\bm{Z}^{id}$. Second, Figure (\ref{fig:g_3_2}) shows that the fused item representation $\bm{Z}$ achieves superior recommendation performance compared to any single aligned modality $\hat{\bm{Z}}^t$, $\hat{\bm{Z}}^v$, and $\bm{Z}^{id}$. This observation is consistent with the conclusion of Theorem \ref{theorem: 1} (2), which supports the comprehensiveness of the unified representation.


\begin{figure}[t]
\centering

\begin{subfigure}[b]{0.5\textwidth}
\centering
\includegraphics[width=1.03\linewidth, height=3.2cm]{fig/g_3_1.png}
\caption{Consistency: Triangular heatmap of performance with different modality weights in Eq. (\ref{equation:16}).}
\label{fig:g_3_1}
\end{subfigure}

\vspace{0.2cm}

\begin{subfigure}[b]{0.5\textwidth}
\centering
\includegraphics[width=0.99\linewidth, height=3.2cm]{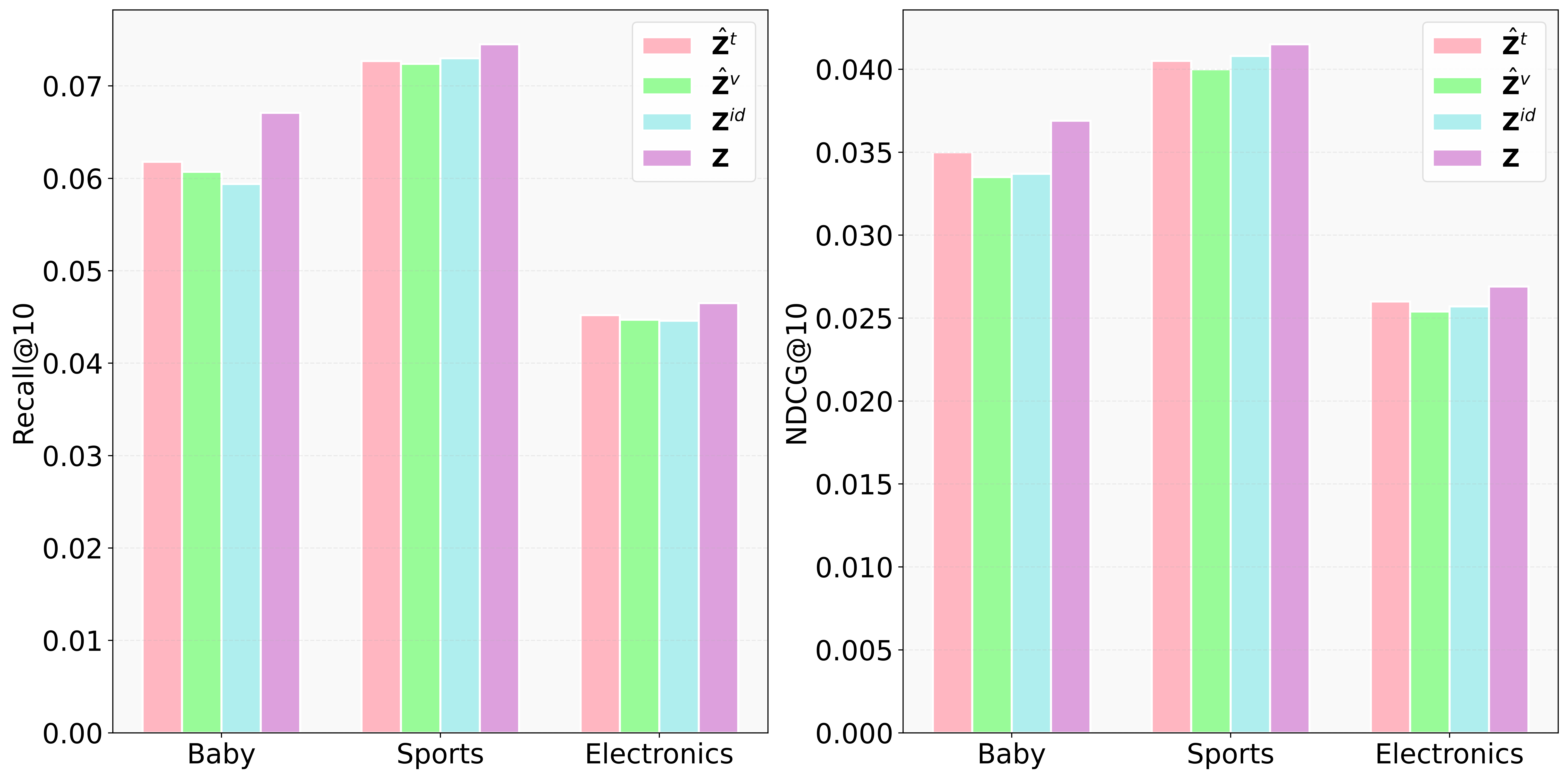}
\caption{Comprehensiveness: Performance comparison of different modalities and the fused representation after alignment,}
\label{fig:g_3_2}
\end{subfigure}

\caption{Alignment Consistency and Fusion Comprehensiveness of RecGOAT on the Baby Dataset.}
\label{fig:g3}
\end{figure}

\subsection{Online Performance (RQ4)}

To verify the performance of RecGOAT in industrial systems, we deployed the OT component of RecGOAT to a large-scale industrial advertising system and conducted rigorous online A/B testing.

Our online baseline ranking model adopts an industry-proven architecture, comprising manually designed features, user sequence encoders, and a stacked deep neural network (DNN). Specifically, the baseline model incorporates two distinct high-dimensional vector representations: a multimodal content embedding \(\bm{z}_{i}^{mm}\) derived from the MLLM, and a collaborative signal representation \(\bm{z}_{i}^{id}\) obtained by mapping recommendation item IDs to embedding vectors.

In the OT component, these vectors \(\bm{z}_{i}^{mm}\), \(\bm{z}_{i}^{id}\) are fed into two DNN layers for nonlinear mapping, generating transformed representations $f_{mm}(\bm{z}_{i}^{mm}) \in \mathbb{R}^d$ and $f_{id}(\bm{z}_{i}^{id}) \in \mathbb{R}^d$. Following \cite{peyre2019computationalOT}, we model $f_{mm}(\bm{z}_{i}^{mm})$ and $f_{id}(\bm{z}_{i}^{id})$ as  Gaussian distributions, denoted as $\mathcal{N}(\bm{\mu_{mm}}, \boldsymbol{\Sigma}_{mm})$ and $\mathcal{N}(\bm{\mu_{id}}, \bm{\Sigma}_{id})$ respectively. The transport cost between the two corresponding vectors can be derived as:

\begin{equation}
\mathcal{W}_2^2(P^{mm}, Q^{id}) = \|\bm{\mu}_{mm} - \bm{\mu}_{id}\|^2 + \mathcal{B}(\boldsymbol{\Sigma}_{mm}, \boldsymbol{\Sigma}_{id})^2, 
\end{equation}

where 

\begin{equation}
\mathcal{B}(\boldsymbol{\Sigma}_{mm}, \boldsymbol{\Sigma}_{id})^2 \stackrel{\text{def}}{=} \operatorname{tr}\left(\boldsymbol{\Sigma}_{mm} + \boldsymbol{\Sigma}_{id} - 2\left(\boldsymbol{\Sigma}_{mm}^{1/2} \boldsymbol{\Sigma}_{id} \boldsymbol{\Sigma}_{mm}^{1/2}\right)^{1/2}\right).
\end{equation}

We then minimize $\mathcal{W}_2^2(P^{mm}, Q^{id})$ to achieve distributional alignment between content embedding and collaborative signal representation in the feature space. Both refined representations are subsequently integrated into the online model as dense features.

We conduct the online A/B test on 5\% traffic from the production system, covering approximately 20 million unique users. As shown in Table~\ref{tab:online_result}, our proposed RecGOAT yields a 1.5\% lift in advertiser value (ADVV) \cite{chai2025longer} relative to the baseline. Notably, RecGOAT achieves a 2.3\% ADVV lift on long-tail data, which demonstrates the superiority of our approach in terms of generalization capability.

\begin{table}[h]
\renewcommand{\arraystretch}{1.0}
\setlength{\tabcolsep}{6.8pt}
\centering
\small 
\caption{Results on online advertising platform.}
\label{tab:online_result}
{
    \begin{tabular}{lcc}
    \toprule
    Method & Setting & ADVV \\
    \midrule
    \multirow{2}{*}{RecGOAT} & all & \textcolor{mypink}{$\uparrow$ 1.5\%} \\
    & long-tail & \textcolor{mypink}{$\uparrow$ 2.3\%} \\
    \bottomrule
    \end{tabular}
}
\end{table}

\section{Hyperparameter Sensitivity Analysis (RQ5)}
To evaluate the sensitivity of RecGOAT to hyperparameters, we report its recommendation performance across three datasets under varying values of $K$ (used for constructing the item-item graph via KNN), as illustrated in Figure \ref{g4}. As observed, when $K$ varies from 10 to 50, the performance metrics (i.e., Recall@10 and NDCG@10) remain highly stable across all datasets. This demonstrates that our proposed model is highly robust to the selection of $K$, thereby alleviating the necessity for exhaustive and time-consuming hyperparameter tuning in practical deployments.

\begin{figure}[h]
    \centering
    \includegraphics[width=3.5in]{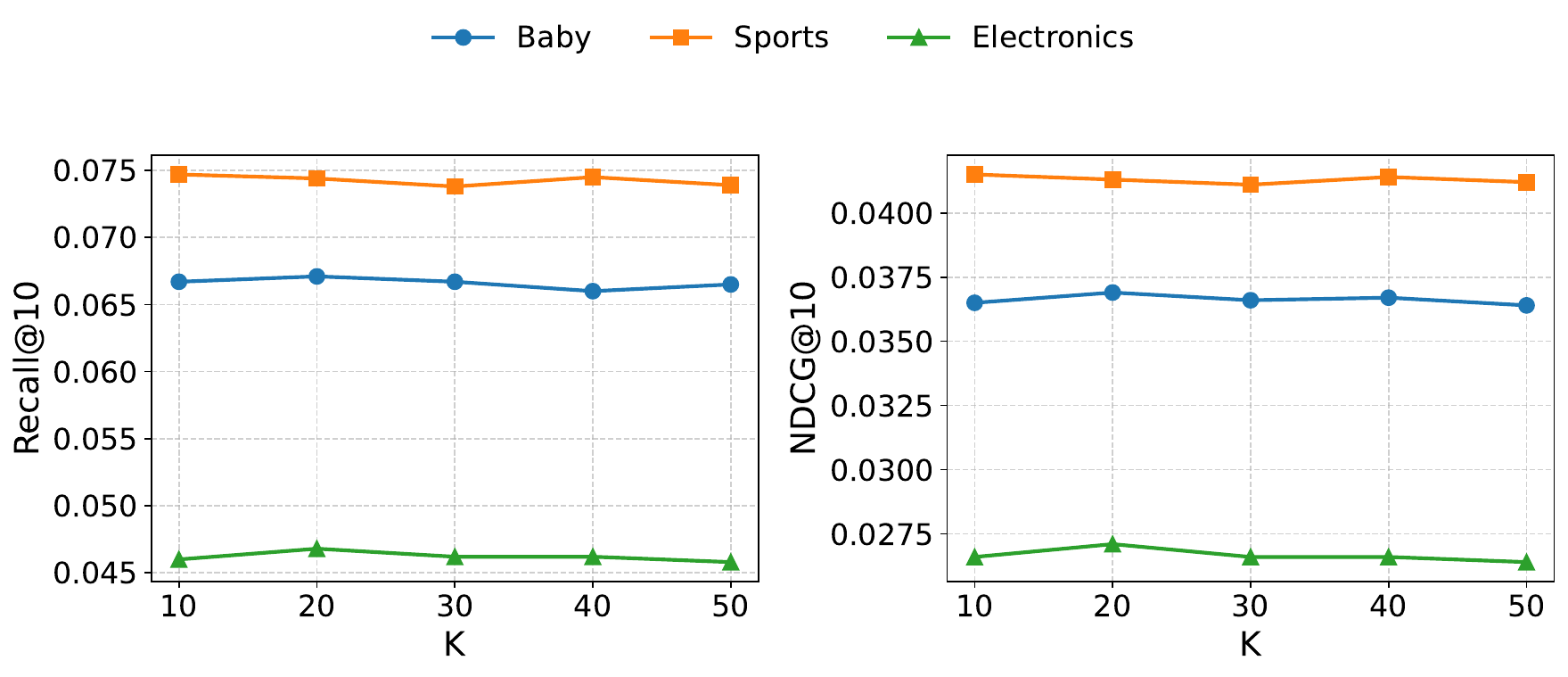}
    \caption{Hyperparameter sensitivity on three Amazon Dataset for the number of nearest neighbors $K$.} 
    \label{g4}
\end{figure}

\section{Runtime Analysis (RQ6)}

To evaluate computational efficiency, we report the training and inference times of RecGOAT alongside three representative baselines on the Baby dataset, as summarized in Table \ref{tab:time_complexity}. Specifically, we select one state-of-the-art and recent model from each category of recommendation paradigms: LightGCN (Traditional ID-based Methods), FindRec (Multimodal Methods), and IRLLRec (LM-enhanced Methods). All models are evaluated under identical hyperparameters and computational resources for a fair comparison. As expected, LightGCN achieves the fastest training and inference speeds owing to its highly simplified linear graph convolutional architecture. Compared to more complex baselines, our proposed RecGOAT outperforms IRLLRec in both training time (3.12 s/epoch) and inference time (1.37 s/evaluation set). Although RecGOAT incurs slightly higher time costs than FindRec, the overall computational overhead remains within the same order of magnitude and is highly acceptable. This indicates that while introducing superior recommendation capabilities, RecGOAT maintains highly competitive computational efficiency, demonstrating its strong potential for deployment in real-world recommender systems.

\begin{table}[t]
\renewcommand{\arraystretch}{1.2}
\setlength{\tabcolsep}{2.8pt}
  \centering
  \caption{Comparison of training and inference time on the Baby Dataset.}
  \label{tab:time_complexity}
  \begin{tabular}{lcccc}
    \toprule
    Model & LightGCN & FindRec & IRLLRec & RecGOAT \\
    \midrule
    \makecell[l]{Training Time \\ (s/training\_epoch)} & 1.33 & 2.95 & 3.41 & 3.12 \\
    \hdashline \addlinespace[0.5ex] 
    \makecell[l]{Inference Time \\ (s/evaluation\_set)} & 0.98 & 1.23 & 2.03 & 1.37 \\
    \bottomrule
  \end{tabular}
\end{table}

\section{Conclusions}

In this paper, we propose RecGOAT, a dual‑granularity semantic alignment framework for LLM‑enhanced multimodal recommendation. It integrates instance‑level alignment via cross‑modal contrastive learning and distribution‑level alignment via optimal adaptive transport to resolve the semantic heterogeneity between large‑model representations and recommendation ID signals. Theoretically, we prove the consistency and comprehensiveness of the aligned representations derived from RecGOAT. Extensive experiments on three Amazon datasets validate our theoretical results and demonstrate SOTA performance against relevant baselines. Furthermore, A/B testing on a large‑scale advertising platform confirms the scalability of RecGOAT. In future work, we will explore interactions among multiple optimal transport alignments and extend semantic alignment solutions to omni-modal large recommendation model.

\bibliographystyle{ACM-Reference-Format}
\bibliography{sample-base}

@String{Computing = "Computing" }

@String{Springer = "Springer-Verlag" }

@inproceedings{he2017NCF,
  title={Neural collaborative filtering},
  author={He, Xiangnan and Liao, Lizi and Zhang, Hanwang and Nie, Liqiang and Hu, Xia and Chua, Tat-Seng},
  booktitle={Proceedings of the 26th international conference on world wide web},
  pages={173--182},
  year={2017}
}

@inproceedings{wang2019NGCF,
  title={Neural graph collaborative filtering},
  author={Wang, Xiang and He, Xiangnan and Wang, Meng and Feng, Fuli and Chua, Tat-Seng},
  booktitle={Proceedings of the 42nd international ACM SIGIR conference on Research and development in Information Retrieval},
  pages={165--174},
  year={2019}
}

@inproceedings{xia2022HCCF,
  title={Hypergraph contrastive collaborative filtering},
  author={Xia, Lianghao and Huang, Chao and Xu, Yong and Zhao, Jiashu and Yin, Dawei and Huang, Jimmy},
  booktitle={Proceedings of the 45th International ACM SIGIR conference on research and development in information retrieval},
  pages={70--79},
  year={2022}
}

@article{zhao2024recommenderinllm,
  title={Recommender systems in the era of large language models (llms)},
  author={Zhao, Zihuai and Fan, Wenqi and Li, Jiatong and Liu, Yunqing and Mei, Xiaowei and Wang, Yiqi and Wen, Zhen and Wang, Fei and Zhao, Xiangyu and Tang, Jiliang and others},
  journal={IEEE Transactions on Knowledge and Data Engineering},
  volume={36},
  number={11},
  pages={6889--6907},
  year={2024},
  publisher={IEEE}
}

@article{lin2025rsbllm,
  title={How can recommender systems benefit from large language models: A survey},
  author={Lin, Jianghao and Dai, Xinyi and Xi, Yunjia and Liu, Weiwen and Chen, Bo and Zhang, Hao and Liu, Yong and Wu, Chuhan and Li, Xiangyang and Zhu, Chenxu and others},
  journal={ACM Transactions on Information Systems},
  volume={43},
  number={2},
  pages={1--47},
  year={2025},
  publisher={ACM New York, NY}
}

@inproceedings{chen2019VECF,
  title={Personalized fashion recommendation with visual explanations based on multimodal attention network: Towards visually explainable recommendation},
  author={Chen, Xu and Chen, Hanxiong and Xu, Hongteng and Zhang, Yongfeng and Cao, Yixin and Qin, Zheng and Zha, Hongyuan},
  booktitle={Proceedings of the 42nd international ACM SIGIR conference on research and development in information retrieval},
  pages={765--774},
  year={2019}
}

@inproceedings{zhou2023bm3,
  title={Bootstrap latent representations for multi-modal recommendation},
  author={Zhou, Xin and Zhou, Hongyu and Liu, Yong and Zeng, Zhiwei and Miao, Chunyan and Wang, Pengwei and You, Yuan and Jiang, Feijun},
  booktitle={Proceedings of the ACM web conference 2023},
  pages={845--854},
  year={2023}
}

@inproceedings{li2025MoDiCF,
  title={Generating with fairness: A modality-diffused counterfactual framework for incomplete multimodal recommendations},
  author={Li, Jin and Wang, Shoujin and Zhang, Qi and Yu, Shui and Chen, Fang},
  booktitle={Proceedings of the ACM on Web Conference 2025},
  pages={2787--2798},
  year={2025}
}

@article{xu2025MMrecsurvey,
  title={A Survey on Multimodal Recommender Systems: Recent Advances and Future Directions},
  author={Xu, Jinfeng and Chen, Zheyu and Yang, Shuo and Li, Jinze and Wang, Wei and Hu, Xiping and Hoi, Steven and Ngai, Edith},
  journal={arXiv preprint arXiv:2502.15711},
  year={2025}
}

@inproceedings{he2016vbpr,
  title={VBPR: visual bayesian personalized ranking from implicit feedback},
  author={He, Ruining and McAuley, Julian},
  booktitle={Proceedings of the AAAI conference on artificial intelligence},
  volume={30},
  number={1},
  year={2016}
}

@article{cui2018mvrnn,
  title={MV-RNN: A multi-view recurrent neural network for sequential recommendation},
  author={Cui, Qiang and Wu, Shu and Liu, Qiang and Zhong, Wen and Wang, Liang},
  journal={IEEE Transactions on Knowledge and Data Engineering},
  volume={32},
  number={2},
  pages={317--331},
  year={2018},
  publisher={IEEE}
}

@inproceedings{pomo2025VLRec,
  title={Do Recommender Systems Really Leverage Multimodal Content? A Comprehensive Analysis on Multimodal Representations for Recommendation},
  author={Pomo, Claudio and Attimonelli, Matteo and Danese, Danilo and Narducci, Fedelucio and Di Noia, Tommaso},
  booktitle={Proceedings of the 34th ACM International Conference on Information and Knowledge Management},
  pages={2377--2387},
  year={2025}
}

@article{malitesta2025modalfusion,
  title={Formalizing multimedia recommendation through multimodal deep learning},
  author={Malitesta, Daniele and Cornacchia, Giandomenico and Pomo, Claudio and Merra, Felice Antonio and Di Noia, Tommaso and Di Sciascio, Eugenio},
  journal={ACM Transactions on Recommender Systems},
  volume={3},
  number={3},
  pages={1--33},
  year={2025},
  publisher={ACM New York, NY}
}

@article{gao2023GNNRec,
  title={A survey of graph neural networks for recommender systems: Challenges, methods, and directions},
  author={Gao, Chen and Zheng, Yu and Li, Nian and Li, Yinfeng and Qin, Yingrong and Piao, Jinghua and Quan, Yuhan and Chang, Jianxin and Jin, Depeng and He, Xiangnan and others},
  journal={ACM Transactions on Recommender Systems},
  volume={1},
  number={1},
  pages={1--51},
  year={2023},
  publisher={ACM New York, NY, USA}
}

@article{anand2025GNNRec,
  title={A survey on recommender systems using graph neural network},
  author={Anand, Vineeta and Maurya, Ashish Kumar},
  journal={ACM Transactions on Information Systems},
  volume={43},
  number={1},
  pages={1--49},
  year={2025},
  publisher={ACM New York, NY, USA}
}

@inproceedings{he2020lightgcn,
  title={Lightgcn: Simplifying and powering graph convolution network for recommendation},
  author={He, Xiangnan and Deng, Kuan and Wang, Xiang and Li, Yan and Zhang, Yongdong and Wang, Meng},
  booktitle={Proceedings of the 43rd International ACM SIGIR conference on research and development in Information Retrieval},
  pages={639--648},
  year={2020}
}

@inproceedings{zhou2023freedom,
  title={A tale of two graphs: Freezing and denoising graph structures for multimodal recommendation},
  author={Zhou, Xin and Shen, Zhiqi},
  booktitle={Proceedings of the 31st ACM international conference on multimedia},
  pages={935--943},
  year={2023}
}

@inproceedings{lin2024gume,
  title={GUME: Graphs and User Modalities Enhancement for Long-Tail Multimodal Recommendation},
  author={Lin, Guojiao and Zhen, Meng and Wang, Dongjie and Long, Qingqing and Zhou, Yuanchun and Xiao, Meng},
  booktitle={Proceedings of the 33rd ACM International Conference on Information and Knowledge Management},
  pages={1400--1409},
  year={2024}
}

@inproceedings{li2023RecFormer,
  title={Text is all you need: Learning language representations for sequential recommendation},
  author={Li, Jiacheng and Wang, Ming and Li, Jin and Fu, Jinmiao and Shen, Xin and Shang, Jingbo and McAuley, Julian},
  booktitle={Proceedings of the 29th ACM SIGKDD Conference on Knowledge Discovery and Data Mining},
  pages={1258--1267},
  year={2023}
}

@inproceedings{yi2024UGT,
  title={A unified graph transformer for overcoming isolations in multi-modal recommendation},
  author={Yi, Zixuan and Ounis, Iadh},
  booktitle={Proceedings of the 18th ACM Conference on Recommender Systems},
  pages={518--527},
  year={2024}
}

@article{yang2023DreamRec,
  title={Generate what you prefer: Reshaping sequential recommendation via guided diffusion},
  author={Yang, Zhengyi and Wu, Jiancan and Wang, Zhicai and Wang, Xiang and Yuan, Yancheng and He, Xiangnan},
  journal={Advances in Neural Information Processing Systems},
  volume={36},
  pages={24247--24261},
  year={2023}
}

@inproceedings{jiang2024diffmm,
  title={Diffmm: Multi-modal diffusion model for recommendation},
  author={Jiang, Yangqin and Xia, Lianghao and Wei, Wei and Luo, Da and Lin, Kangyi and Huang, Chao},
  booktitle={Proceedings of the 32nd ACM International Conference on Multimedia},
  pages={7591--7599},
  year={2024}
}

@inproceedings{wang2025findrec,
  title={FindRec: Stein-Guided Entropic Flow for Multi-Modal Sequential Recommendation},
  author={Wang, Maolin and Xiao, Yutian and Wang, Binhao and Zhang, Sheng and Ye, Shanshan and Wang, Wanyu and Yin, Hongzhi and Guo, Ruocheng and Xu, Zenglin},
  booktitle={Proceedings of the 31st ACM SIGKDD Conference on Knowledge Discovery and Data Mining V. 2},
  pages={3008--3018},
  year={2025}
}

@inproceedings{bao2023tallrec,
  title={Tallrec: An effective and efficient tuning framework to align large language model with recommendation},
  author={Bao, Keqin and Zhang, Jizhi and Zhang, Yang and Wang, Wenjie and Feng, Fuli and He, Xiangnan},
  booktitle={Proceedings of the 17th ACM conference on recommender systems},
  pages={1007--1014},
  year={2023}
}

@inproceedings{
wei2024UniMP,
title={Towards Unified Multi-Modal Personalization: Large Vision-Language Models for Generative Recommendation and Beyond},
author={Tianxin Wei and Bowen Jin and Ruirui Li and Hansi Zeng and Zhengyang Wang and Jianhui Sun and Qingyu Yin and Hanqing Lu and Suhang Wang and Jingrui He and Xianfeng Tang},
booktitle={The Twelfth International Conference on Learning Representations},
year={2024},
url={https://openreview.net/forum?id=khAE1sTMdX}
}

@inproceedings{yi2025GollaRec,
  title={A Multi-modal Large Language Model with Graph-of-Thought for Effective Recommendation},
  author={Yi, Zixuan and Ounis, Iadh},
  booktitle={Proceedings of the 2025 Conference of the Nations of the Americas Chapter of the Association for Computational Linguistics: Human Language Technologies (Volume 1: Long Papers)},
  pages={1591--1606},
  year={2025}
}

@inproceedings{wei2019mmgcn,
  title={MMGCN: Multi-modal graph convolution network for personalized recommendation of micro-video},
  author={Wei, Yinwei and Wang, Xiang and Nie, Liqiang and He, Xiangnan and Hong, Richang and Chua, Tat-Seng},
  booktitle={Proceedings of the 27th ACM international conference on multimedia},
  pages={1437--1445},
  year={2019}
}

@inproceedings{zhang2021LATTICE,
  title={Mining latent structures for multimedia recommendation},
  author={Zhang, Jinghao and Zhu, Yanqiao and Liu, Qiang and Wu, Shu and Wang, Shuhui and Wang, Liang},
  booktitle={Proceedings of the 29th ACM international conference on multimedia},
  pages={3872--3880},
  year={2021}
}

@article{vaswani2017attention,
  title={Attention is all you need},
  author={Vaswani, Ashish and Shazeer, Noam and Parmar, Niki and Uszkoreit, Jakob and Jones, Llion and Gomez, Aidan N and Kaiser, {\L}ukasz and Polosukhin, Illia},
  journal={Advances in neural information processing systems},
  volume={30},
  year={2017}
}

@inproceedings{gu2024mamba,
  title={Mamba: Linear-time sequence modeling with selective state spaces},
  author={Gu, Albert and Dao, Tri},
  booktitle={First conference on language modeling},
  year={2024}
}

@article{lopez2025LLM4MRSsurvey,
  title={A Survey on Large Language Models in Multimodal Recommender Systems},
  author={Lopez-Avila, Alejo and Du, Jinhua},
  journal={arXiv preprint arXiv:2505.09777},
  year={2025}
}

@article{liu2024Rec-gpt4v,
  title={Rec-gpt4v: Multimodal recommendation with large vision-language models},
  author={Liu, Yuqing and Wang, Yu and Sun, Lichao and Yu, Philip S},
  journal={arXiv preprint arXiv:2402.08670},
  year={2024}
}

@inproceedings{zhang2025notellm2,
  title={NoteLLM-2: Multimodal Large Representation Models for Recommendation},
  author={Zhang, Chao and Zhang, Haoxin and Wu, Shiwei and Wu, Di and Xu, Tong and Zhao, Xiangyu and Gao, Yan and Hu, Yao and Chen, Enhong},
  booktitle={31st ACM SIGKDD Conference on Knowledge Discovery and Data Mining (KDD 2025)},
  pages={2815--2826},
  year={2025},
  organization={Association for Computing Machinery}
}

@inproceedings{wang2025IRLLRec,
  title={Intent representation learning with large language model for recommendation},
  author={Wang, Yu and Sang, Lei and Zhang, Yi and Zhang, Yiwen},
  booktitle={Proceedings of the 48th International ACM SIGIR Conference on Research and Development in Information Retrieval},
  pages={1870--1879},
  year={2025}
}

@article{santambrogio2015OTapp,
  title={Optimal transport for applied mathematicians},
  author={Santambrogio, Filippo},
  year={2015},
  publisher={Springer}
}

@article{peyre2025OTML,
  title={Optimal Transport for Machine Learners},
  author={Peyr{\'e}, Gabriel},
  journal={arXiv preprint arXiv:2505.06589},
  year={2025}
}

@inproceedings{
cao2022otkge,
title={{OTKGE}: Multi-modal Knowledge Graph Embeddings via Optimal Transport},
author={Zongsheng Cao and Qianqian Xu and Zhiyong Yang and Yuan He and Xiaochun Cao and Qingming Huang},
booktitle={Advances in Neural Information Processing Systems},
year={2022},
url={https://openreview.net/forum?id=gbXqMdxsZIP}
}

@inproceedings{chen2020got,
  title={Graph optimal transport for cross-domain alignment},
  author={Chen, Liqun and Gan, Zhe and Cheng, Yu and Li, Linjie and Carin, Lawrence and Liu, Jingjing},
  booktitle={International Conference on Machine Learning},
  pages={1542--1553},
  year={2020},
  organization={PMLR}
}

@inproceedings{yang2023MOTKD,
  title={Multimodal optimal transport knowledge distillation for cross-domain recommendation},
  author={Yang, Wei and Yang, Jie and Liu, Yuan},
  booktitle={Proceedings of the 32nd ACM International Conference on Information and Knowledge Management},
  pages={2959--2968},
  year={2023}
}

@article{zhang2025qwen3embedding,
  title={Qwen3 Embedding: Advancing Text Embedding and Reranking Through Foundation Models},
  author={Zhang, Yanzhao and Li, Mingxin and Long, Dingkun and Zhang, Xin and Lin, Huan and Yang, Baosong and Xie, Pengjun and Yang, An and Liu, Dayiheng and Lin, Junyang and others},
  journal={arXiv preprint arXiv:2506.05176},
  year={2025}
}

@article{liu2023llava,
  title={Visual instruction tuning},
  author={Liu, Haotian and Li, Chunyuan and Wu, Qingyang and Lee, Yong Jae},
  journal={Advances in neural information processing systems},
  volume={36},
  pages={34892--34916},
  year={2023}
}

@inproceedings{velivckovic2018gat,
  title={Graph Attention Networks},
  author={Veli{\v{c}}kovi{\'c}, Petar and Cucurull, Guillem and Casanova, Arantxa and Romero, Adriana and Li{\`o}, Pietro and Bengio, Yoshua},
  booktitle={International Conference on Learning Representations},
  year={2018}
}

@article{yang2025qwen3,
  title={Qwen3 technical report},
  author={Yang, An and Li, Anfeng and Yang, Baosong and Zhang, Beichen and Hui, Binyuan and Zheng, Bo and Yu, Bowen and Gao, Chang and Huang, Chengen and Lv, Chenxu and others},
  journal={arXiv preprint arXiv:2505.09388},
  year={2025}
}

@inproceedings{chen2020simcl,
  title={A simple framework for contrastive learning of visual representations},
  author={Chen, Ting and Kornblith, Simon and Norouzi, Mohammad and Hinton, Geoffrey},
  booktitle={International conference on machine learning},
  pages={1597--1607},
  year={2020},
  organization={PmLR}
}

@inproceedings{li2024cdnmf,
  title={Contrastive deep nonnegative matrix factorization for community detection},
  author={Li, Yuecheng and Chen, Jialong and Chen, Chuan and Yang, Lei and Zheng, Zibin},
  booktitle={ICASSP 2024-2024 IEEE International Conference on Acoustics, Speech and Signal Processing (ICASSP)},
  pages={6725--6729},
  year={2024},
  organization={IEEE}
}

@article{peyre2019computationalOT,
  title={Computational optimal transport: With applications to data science},
  author={Peyr{\'e}, Gabriel and Cuturi, Marco and others},
  journal={Foundations and Trends{\textregistered} in Machine Learning},
  volume={11},
  number={5-6},
  pages={355--607},
  year={2019},
  publisher={Now Publishers, Inc.}
}

@article{sinkhorn1967first,
  title={Concerning nonnegative matrices and doubly stochastic matrices},
  author={Sinkhorn, Richard and Knopp, Paul},
  journal={Pacific Journal of Mathematics},
  volume={21},
  number={2},
  pages={343--348},
  year={1967},
  publisher={Mathematical Sciences Publishers}
}

@article{cuturi2013sinkhorn,
  title={Sinkhorn distances: Lightspeed computation of optimal transport},
  author={Cuturi, Marco},
  journal={Advances in neural information processing systems},
  volume={26},
  year={2013}
}

@article{courty2017OTproof,
  title={Joint distribution optimal transportation for domain adaptation},
  author={Courty, Nicolas and Flamary, R{\'e}mi and Habrard, Amaury and Rakotomamonjy, Alain},
  journal={Advances in neural information processing systems},
  volume={30},
  year={2017}
}

@inproceedings{rendle2009bpr,
  title={BPR: Bayesian personalized ranking from implicit feedback},
  author={Rendle, Steffen and Freudenthaler, Christoph and Gantner, Zeno and Schmidt-Thieme, Lars},
  booktitle={Proceedings of the Twenty-Fifth Conference on Uncertainty in Artificial Intelligence},
  pages={452--461},
  year={2009}
}

@inproceedings{mcauley2015amazon,
  title={Image-based recommendations on styles and substitutes},
  author={McAuley, Julian and Targett, Christopher and Shi, Qinfeng and Van Den Hengel, Anton},
  booktitle={Proceedings of the 38th international ACM SIGIR conference on research and development in information retrieval},
  pages={43--52},
  year={2015}
}

@inproceedings{xu2025fastMMrec,
  title={The Best is Yet to Come: Graph Convolution in the Testing Phase for Multimodal Recommendation},
  author={Xu, Jinfeng and Chen, Zheyu and Yang, Shuo and Li, Jinze and Ngai, Edith CH},
  booktitle={Proceedings of the 33rd ACM International Conference on Multimedia},
  pages={6325--6334},
  year={2025}
}

@inproceedings{zhou2023mmrec,
  title={Mmrec: Simplifying multimodal recommendation},
  author={Zhou, Xin},
  booktitle={Proceedings of the 5th ACM International Conference on Multimedia in Asia Workshops},
  pages={1--2},
  year={2023}
}

@inproceedings{kim2024allmrec,
  title={Large language models meet collaborative filtering: An efficient all-round llm-based recommender system},
  author={Kim, Sein and Kang, Hongseok and Choi, Seungyoon and Kim, Donghyun and Yang, Minchul and Park, Chanyoung},
  booktitle={Proceedings of the 30th ACM SIGKDD Conference on Knowledge Discovery and Data Mining},
  pages={1395--1406},
  year={2024}
}

@inproceedings{chai2025longer,
  title={Longer: Scaling up long sequence modeling in industrial recommenders},
  author={Chai, Zheng and Ren, Qin and Xiao, Xijun and Yang, Huizhi and Han, Bo and Zhang, Sijun and Chen, Di and Lu, Hui and Zhao, Wenlin and Yu, Lele and others},
  booktitle={Proceedings of the Nineteenth ACM Conference on Recommender Systems},
  pages={247--256},
  year={2025}
}

@inproceedings{yang2024mami,
  title={Multimodal-aware Multi-intention Learning for Recommendation},
  author={Yang, Wei and Yang, Qingchen},
  booktitle={Proceedings of the 32nd ACM International Conference on Multimedia},
  pages={5663--5672},
  year={2024}
}

@inproceedings{yang2025fitmm,
  title={FITMM: Adaptive Frequency-Aware Multimodal Recommendation via Information-Theoretic Representation Learning},
  author={Yang, Wei and Zhong, Rui and Chen, Yiqun and Li, Shixuan and Ping, Heng and Lu, Chi and Jiang, Peng},
  booktitle={Proceedings of the 33rd ACM International Conference on Multimedia},
  pages={6193--6202},
  year={2025}
}







\end{document}